\documentclass[11pt]{article}
\usepackage{times}
\usepackage{fullpage}
\usepackage{epsf}
\usepackage{amsmath,amsfonts}
\usepackage{lineno}
\usepackage{xcolor}
\usepackage{bm}
\usepackage{graphicx}
\usepackage{complexity}

\usepackage[ruled,vlined,commentsnumbered,titlenotnumbered]{algorithm2e}
\SetKwBlock{Loop}{loop}{}
\def\asgn{\leftarrow}

\newcommand\unfrac\frac

\newcommand\cutcolor{\texttt{cut-or-color}}
\newcommand\cutcolornew{\texttt{cut-or-color}}
\newcommand\monochromatic{\texttt{monochromatic}}
\newcommand\sidecut{\texttt{best-side-cut}}

\def\regul{\texttt{regularize}}
\newcommand\avg[1]{\textnormal{avg}\,{#1}}

\newcommand\wt{\widetilde}
\newcommand\tO{{\widetilde O}\/}
\newcommand\tto{{\widetilde o}}
\newcommand\tomega{{\widetilde \omega}}
\newcommand\tOmega{{\widetilde\Omega}}

\newcommand\ds{\Delta}
\newcommand\dt{\delta}
\newcommand\nh{\Psi}
\newcommand\tcite[1]{~\cite{#1}}
\newcommand\tref[1]{~\ref{#1}}

\newtheorem{proposition}{Proposition}
\newtheorem{lemma}[proposition]{Lemma}

\newtheorem{theorem}[proposition]{Theorem}

\newcommand{\qed}{\hbox{\rule{6pt}{6pt}}}
\newenvironment{proof}[1][]{\paragraph{Proof{#1}}}{\hfill\qed\medskip\\}

\newcommand\drop[1]{}

\newcommand\req[1]{(\ref{#1})}

\newcommand\packlist{\setlength{\itemsep}{1pt}
\setlength{\parskip}{0pt}\setlength{\parsep}{0pt}}

\newcounter{invari}
\newenvironment{invariants}{\begin{enumerate}\packlist\setcounter{enumi}{\value{invari}}
\renewcommand\theenumi{(\roman{enumi})}
\renewcommand\labelenumi\theenumi}{\setcounter{invari}{\value{enumi}}\end{enumerate}}

\begin{document}
\setcounter{page}0

%
\title{Better coloring of 3-colorable graphs}

\author{{\em Ken-ichi Kawarabayashi}\thanks{Ken-ichi Kawarabayashi's
   research is partly supported by JSPS Kakenhi 22H05001, by JSPS Kakenhi JP20A402 and by JST ASPIRE JPMJAP2302.}\\National Institute of Informatics \& The University of Tokyo, Tokyo, Japan\\
  \texttt{k\_keniti@nii.ac.jp}
  \and {\em Mikkel Thorup}\thanks{Mikkel Thorup's research supported by VILLUM Foundation grant 16582, Basic Algorithms Research Copenhagen (BARC).}\\
  University of Copenhagen, Denmark\\
 \texttt{mikkel2thorup@gmail.com}
 \and {\em Hirotaka Yoneda}\thanks{Hirotaka Yoneda's 
   research is partly supported by JSPS Kakenhi 22H05001 and by JST ASPIRE JPMJAP2302.}\\The University of Tokyo, Tokyo, Japan\\
  \texttt{squar37@gmail.com}}

\maketitle
\thispagestyle{empty}

\begin{abstract}
We consider the problem of coloring a 3-colorable graph in polynomial
time using as few colors as possible. 
This is one of the most challenging problems in graph algorithms. 
In this paper using Blum's notion of ``progress'', we develop a new combinatorial algorithm for the following: Given any 3-colorable graph with minimum degree $\ds>\sqrt n$, 
we can, in polynomial time, make progress towards a $k$-coloring for some $k=\sqrt{n/\ds}\cdot n^{o(1)}$. 

We balance our main result with the best-known semi-definite(SDP) approach which we 
use for degrees below
$n^{0.605073}$. As a
result, we show that $\tO(n^{0.19747})$ colors suffice for coloring 3-colorable graphs. This improves on the previous best bound of $\tO(n^{0.19996})$ by Kawarabayashi and Thorup \cite{KT17}. 
\end{abstract}

\newpage
\section{Introduction}
Recognizing 3-colorable graphs is a classic problem, which was proved to be NP-hard by
Garey, Johnson, and Stockmeyer \cite{GJS76} in 1976,
and was one of the
prime examples of NP-hardness mentioned by Karp in 1975
\cite{Karp75}. In this paper, we also focus on 3-colorable graphs.

Given a 3-colorable graph,
that is a graph with an unknown 3-coloring, we try to color it in
polynomial time using as few colors as possible. The algorithm
is allowed to give up if the input graph is not 3-colorable.
If a coloring is produced, we can always check that it is valid even if
the input graph is not 3-colorable. This challenge has
engaged many researchers in theoretical computer science for a long time, yet we are far from understanding it. Wigderson \cite{Wig83} is the first to give a combinatorial algorithm for $O(n^{1/2})$ coloring for a 3-colorable graph with $n$ vertices in 1982. Berger and
Rompel \cite{BR90} then improved this result to $O((n/(\log n))^{1/2})$\footnote{If the base is not specified, $\log$ 
denotes $\log_2$.}. Blum\tcite{Blum94} gave the first polynomial improvements, based on combinatorial algorithms, first to
$\tO(n^{2/5})$ colors in 1989, and then to $\wt O(n^{3/8})$ colors
in 1990.

The next big step in 1998 was given by Karger, Motwani, and Sudan
\cite{KMS98} using semi-definite programming (SDP). This improvement uses Goemans and Williamson's seminal SDP-based algorithm for the max-cut problem in 1995 \cite{GW95}.
Formally, for a graph with maximum degree $\Delta_{\max}$,
Karger et al. \cite{KMS98} were able to get down to $O(\Delta_{\max}^{1/3})$ colors. Combining
this with Wigderson's combinatorial algorithm, they got down to $\tO(n^{1/4})$ colors.
Later in 1997, Blum and Karger  \cite{BK97} combined the SDP from \cite{KMS98}
with Blum's \cite{Blum94} combinatorial algorithm, yielding an improved bound of $\wt O(n^{3/14}) =\tO(n^{0.2142})$. Later improvements on
SDP have also been combined with Blum's
combinatorial algorithm. In 2006, Arora, Chlamtac, and Charikar \cite{ACC06} got
down to $O(n^{0.2111})$ colors. The proof in \cite{ACC06} is based on
the seminal result of Arora, Rao, and Vazirani \cite{ARV09} 
which gives an $O(\sqrt{\log n})$ approximation algorithm for the sparsest cut problem.
Further progress in SDP is made; Chlamtac \cite{Chl07} got down to $O(n^{0.2072})$
colors. 

In 2012, Kawarabayashi and Thorup \cite{KT12}  made the first improvements on combinatorial algorithms over Blum\tcite{Blum94}; they showed $\wt O(n^{4/11})$, improving over Blum's $\wt O(n^{3/8})$ bounds from his combinatorial algorithm. They also got down to $\tO(n^{0.2049})$ colors by combining with SDP results. Finally, they \cite{KT17}
got down to $\tO(n^{0.19996})$ colors, by further utilizing SDP and the combinatorial algorithm in \cite{KT12}. 

On the lower bound side, for general graphs, the chromatic number is
inapproximable in polynomial time within factor $n^{1-\epsilon}$ for
any constant $\epsilon > 0$, unless
\coRP=\NP\ \cite{FK,Ha}.  This
lower bound is much higher than the above-mentioned upper bounds for
3-colorable graphs.  The lower bounds known for coloring 
3-colorable graphs are much weaker.  We know that it is \NP-hard to get
down to $5$ colors \cite{VGK04,KLS00}. Dinur, Mossel and Regev \cite{DMR09} proved that it is hard to color 3-colorable graphs
with any constant
number of colors (i.e., $O(1)$ colors) based on a variant of the
Unique Games Conjecture. Same hardness was proved based on weaker conjecture in \cite{GS20}. Stronger hardness bounds are known if the graph is
only ``almost'' 3-colorable \cite{DHSV2015}.
Some integrality gap results \cite{FLS04,KMS98,Sze94} show that the
simple SDP relaxation has integrality gap at least $n^{0.157}$ but such
a gap is not known for SDPs using levels of Lasserre lifting
\cite{ACC06,AG11,Chl07}.

\subsection{Interplay  between combinatorial and SDP approaches and our main result}

In this paper, we will further
improve the coloring
bound for 3-colorable graphs to $\tO(n^{0.19747})$ colors.
However, what makes
our results interesting is that we converge towards a limit for known
combinatorial approaches as explained below.

To best describe our own and previous results, we need Blum's\tcite{Blum94} notion of {\em progress towards $k$-coloring}. The basic
idea is that if we for any 3-colorable graph can make
progress towards $k$-coloring in polynomial time, then we can $\tO(k)$-color any 3-colorable graph. 

Reductions from  \cite{ACC06,BK97,KT17} show
that for any parameter $\ds$, it suffices to find progress towards
$k$-coloring for 3-colorable graphs that have either minimum degree $\Delta$ or maximum degree $\Delta$. High minimum degree has been good for combinatorial approaches while low
maximum degree has been good for SDP approaches.  The best bounds are
obtained by choosing $\Delta$ to balance between the best SDP and
combinatorial algorithms. 

\paragraph{Purely combinatorial bounds}
On the combinatorial side,
for 3-colorable graphs with minimum degree $\ds=n^{1-\Omega(1)}$, the previous bounds for progress have followed the following sequence 
\begin{equation}\label{comb}
\tO((n/\ds)^{i/(2i-1)}).
\end{equation}
Here $i=1$ is by Wigderson in STOC'82 (covered by \cite{Wig83}), while $i=2$ is by Blum from STOC'89, and $i=3$ is by
Blum at FOCS'90 (both covered by \cite{Blum94}). Finally, $i=4$ is by Kawarabayashi and Thorup from FOCS'12 (covered by \cite{KT17}). 

Since we can trivially color a graph with maximum
degree $\ds$ with $\ds+1$ colors, \req{comb}
also yields combinatorial coloring bounds with $\tO(n^{i/(3i-1})$ colors assuming
\req{comb} for any $\ds\geq n^{i/(3i-1)}$. Thus
we get the combinatorial $\tO(n^{1/2})$-coloring of Wigderson, the $\tO(n^{2/5})$- and $\tO(n^{3/8})$-colorings of Blum and the 
$\tO(n^{4/11})$-coloring of Kawarabayashi and Thorup.

\paragraph{Combination with SDP}
For 3-colorable graphs with maximum degree $\Delta$, Karger et al. \cite{KMS98} used SDP to get down to $O(\ds^{1/3})$ colors. In
combination with \req{comb} this means that
we can color 3-colorable graphs with
$\tO(n^{i/(5i-1)})$ colors. This gave
them $\tO(n^{1/4})$ colors in combination
with Wigderson's $i=1$, and Blum and Karger 
got $\tO(n^{3/14})$ using Blum's $i=3$ from
\cite{Blum94}. Later improvements in SDP 
\cite{ACC06,Chl07} also got
combined with Blum's result \cite{Blum94}.

\paragraph{Combinatorial algorithm for balance with SDP}
We now note that \req{comb} converges to $(n/\ds)^{1/2}$ from above. Balancing with
a SDP bound which is as good or better than
$\ds^{1/3}$, we would only need \req{comb} 
to hold for $\ds\geq n^{3/5}$. This idea
was used in \cite{KT17} which got \req{comb} to
hold for $i=12$, but only for minimum degree
$\ds\geq n^{0.61674333}$. In combination
with the current best SDP by Chlamtac
\cite{Chl07}, this leads them to an overall coloring bound of $\tO(n^{0.19996})$. 
In principle, \cite{KT17} could prove
\req{comb} for even larger $i$, but then
the minimum degree would also have to be larger, and then the balance with SDP would lead to worse overall bounds.
\paragraph{Our result}
For $i\rightarrow \infty$, the sequence \req{comb} approaches $\tO((n/\ds)^{1/2})$. In this paper, we get arbitrarily close to this limit assuming only that $\Delta\geq\sqrt n$. More precisely, our main 
combinatorial result is:
\begin{theorem}\label{thm:main}
In polynomial time, for any 3-colorable graph with $n$ vertices with minimum degree $\ds>n^{0.5}$,
we can make progress towards a $k$-coloring for some $k=n^{o(1)} (n/\ds)^{1/2}$  
\end{theorem}
For the optimal balance
with the best SDP by Chlamtac \cite{Chl07}, we will
set $\ds=n^{0.605073}$ which
is comfortably bigger
than our limit $\ds>\sqrt{n}$.
Thereby we improve the best-known bound of  $\tO(n^{0.19996})$ by Kawarabayashi and Thorup \cite{KT17} to  color any 3-colorable
$n$ vertex graph in polynomial time, as follows.  

\begin{theorem}\label{thm:main1}
In polynomial time, for any 3-colorable graph with $n$ vertices, we can give an $\tO(n^{0.19747})$-coloring. 
\end{theorem}

To appreciate the simple
degree bound $\ds>n^{0.5}$ from Theorem \ref{thm:main},
we state here the corresponding result from
\cite{KT17}:
\begin{theorem}[{\cite[Theorem 23]{KT17}}]
Consider a 3-colorable graph on $n$ vertices with all degrees above $\ds$ where
$\ds=n^{1-\Omega(1)}$.
Suppose for some integer $c=O(1)$ that
$
k=\tomega\left((n/\ds)^{\frac{2c+2}{4c+3}}\right)$
and for all $j=1,...,c-1$,
\begin{equation*}\label{eq:degree17}
(\ds/k)(\ds k/n)^j(\ds k^2/n)^{j(j+1)}=
\ds^{j^2+2j+1}k^{\,2j^2+3j-1}/n^{\,j^2+2j}=
\tomega(1). \end{equation*}
Then we can make progress
towards $\tO(k)$ coloring 
in polynomial time.
\end{theorem}

We note that ${\frac{2c+2}{4c+3}}={\frac{i}{2i-1}}$ for $i=2c+1$, so the coloring bounds from \cite{KT17} still follow
the pattern from \req{comb}. However, the
degree constraint from \cite{KT17} 
is thus both
more complicated and more
restrictive than our $\ds\geq\sqrt n$, limiting
the balancing with SDP. Our Theorem \ref{thm:main} thus provides a both stronger and more appealing understanding.

\paragraph{Techniques}
Our coloring algorithm follows
the same general pattern 
as that in \cite{KT17}, which
recurses through a sequence of nested cuts, called ``sparse cuts", until it finds progress.
Here we go through the same recursion, but in addition to
the sparse cuts, we identify a family of alternative  ``side cuts".
Having the choice between
the side cuts and the sparse cuts is what leads us to
the nicer and stronger bounds from Theorem \ref{thm:main}. To describe our new side cuts, we first have to review the algorithm from \cite{KT17}.

\drop{

Our goal is to improve the overall coloring bound in terms of
$n$, and we will indeed get down to $o(n^{1/5})$ colors. To do so, we
specifically target the connection between combinatorial and
SDP approaches. Using our new combinatorial algorithm
as a subroutine, we present a novel recursion that gets us down to
$\tO((n/\ds)^{12/23})$ colors, but only for the large values
of $\ds$ needed for an optimal combination with
SDP. In combination with Chlamtac's SDP \cite{Chl07}, we get a polynomial
time algorithm that colors any 3-colorable graph on $n$ vertices with
$O(n^{0.19996})$ colors.

We note that for smaller values of $\Delta$, our new recursion does
not offer any improvement over our new combinatorial bound
$\tO((n/\ds)^{4/7})$. Instead of adding another independent dot, the
new recursion connects the dots, improving the combinatorial side only
in the parameter range of relevance for combination with SDP.

We specifically target the worst-case for SDP: high degrees.
}

\section{Preliminaries}\label{sec:blum}

In this section, we provide the notations needed in this paper. They are actually the same as those used in \cite{KT17} (and indeed in \cite{Blum94}), but for completeness, we give here all necessary notations.  

We hide $\log n$ factors, so we
use the notation that $\tO(x)\leq
x\log^{O(1)} (n)$,
$\tOmega(x)\geq x/\log^{O(1)} (n)$, $\tto(x)\leq x/\log^{\omega(1)} (n)$,
and $\tomega(x)\geq x\log^{\omega(1)} (n)$.

We are given a 3-colorable graph
$G=(V,E)$ with $|V|=n=\omega(1)$ vertices.  The (unknown) 3-colorings are 
red, green, and blue.  For a vertex $v$, we let $N(v)$ denote its set
of neighbors.  For a vertex set $X \subseteq V$, let
$N(X)=\bigcup_{v\in X}N(v)$ be the neighborhood of $X$. If $Y$ is a
vertex set, we use $N_Y$ to denote neighbors in $Y$, so
$N_Y(v)=N(v)\cap Y$ and $N_Y(X)=N(X)\cap Y$.
We let $d_Y(v)=|N(v)\cap
Y|$ and $d_Y(X)=\{d_Y(v)\,|\,v\in X\}$.  Then $\min d_Y(X)$, $\max
d_Y(X)$, and $\avg d_Y(X)$, denote the minimum, maximum, and average
degree from $X$ to $Y$.

For some color target $k$ depending on $n$, we wish to find
an $\tO(k)$ coloring of $G$ in polynomial time.  We reuse several ideas and techniques from
Blum's approach\tcite{Blum94}.
\paragraph*{Progress}
Blum has a general notion of {\em progress towards an $\tO(k)$ coloring}
(or {\em progress} for short if $k$ is understood). The basic
idea is that such progress eventually leads to a full $\tO(k)$ coloring
of a graph. Blum presents three types of progress towards $\tO(k)$ coloring:
\begin{description}
\item[Type 0: Same color.] Finding vertices $u$ and $v$ that have
the same color in every  3-coloring.

\item[Type 1: Large independent set.] Finding an independent or 2-colorable vertex set $X$ of size $\tOmega(n/k)$.

\item[Type 2: Small neighborhood.] Finding a non-empty independent
or 2-colorable vertex set $X$ such that $|N(X)|=\tO(k|X|)$.
\end{description}
In order to get from progress to actual coloring, we want
$k$ to be bounded by a
{\em near-polynomial\/} function $f$ of $n$ where
near-polynomial means that $f$ is non-decreasing and that there are
constants $c,c'>1$ such that $c f(n)\leq f(2n)\leq c'
f(n)$ for all $n$. As described in \cite{Blum94}, this includes any function of
the form $f(n)=n^\alpha \log^\beta n$ for constants $\alpha>0$ and
$\beta$.

\begin{lemma}[{\cite[Lemma 1]{Blum94}}]\label{lem:progress} Let $f$ be near-polynomial. If we in time polynomial in $n$
can make progress towards an $\tO(f(n))$  coloring of either Type 0, 1, or 2, on any 3-colorable graph on $n$ vertices,
then in time polynomial in $n$, we can $\tO(f(n))$ color any 3-colorable graph on $n$ vertices.
\end{lemma}

The general strategy is to identify a small
parameter $k$ for which we can guarantee progress. To apply
Lemma \ref{lem:progress} and get a coloring, we need a bound $f$ on $k$
where $f$ is near-polynomial in $n$. As
soon as we find one progress of the above types, we are done, so
generally, whenever we see a condition that implies progress, we
assume that the condition is not satisfied.

Our focus is to find a vertex set $X$, $|X|>1$, that
is guaranteed to be monochromatic in every 3-coloring. This will
happen assuming that we do not make other progress on the way. When we
have the vertex set $X$, we get same-color progress for any pair of
vertices in $X$.  We refer to this as {\em monochromatic
  progress}.

Most of our progress will be made
via the results of Blum presented below using a common parameter
\begin{equation}\label{eq:nh}
\nh = n/k^2.
\end{equation}
A very useful tool we get from Blum is the
following multichromatic (more than one color) test:
\begin{lemma}[{\cite[Corollary 4]{Blum94}}] \label{lem:blum} Given a vertex set $X\subseteq V$ of
size at least $\nh = n/k^2$, in polynomial time, we can either make
progress towards an $\tO(k)$-coloring of $G$, or else guarantee that
under {\em every} legal 3-coloring of $G$, the set $X$ is
multichromatic.
\end{lemma}
The following lemma is implicit in \cite{Blum94} and explicit
in \cite{KT17}.
\begin{lemma}[{\cite[Lemma 6]{KT17}}]\label{lem:large-neighborhoods}
If the vertices in a set $Z$ on average have $d$ neighbors in $U$, then
the whole set $Z$ has at least $\min\{d/\nh,|Z|\}\,d/2$
distinct neighbors in $U$ (otherwise some progress is made).
\end{lemma}

\paragraph*{Large minimum degree}
Our algorithms will exploit a lower bound $\Delta$ on the minimum degree in the
graph.
It is easily seen that if a vertex $v$ has
$d$ neighbors, then we can make progress towards $\tO(d)$ coloring since
this is a small neighborhood for Type 2 progress. For our color target $k$,
we may therefore assume:
\begin{equation}\label{eq:min-degree}
k\leq\Delta/\log^a n\textnormal{ for any constant $a$}.
\end{equation}
However, combined with semi-definite programming (SDP) as in \cite{BK97}, we can assume a much larger
minimum degree. The combination is captured by the following
lemma, which is proved  in \cite{KT17}:
\begin{lemma}[{\cite[Proposition 17]{KT17}}] \label{lem:sdp}
Suppose for some near-polynomial functions $d$ and $f$, that for any $n$, we can make progress towards an $\tO(f(n))$ coloring for
\begin{itemize}
\item any 3-colorable graph on $n$ vertices with minimum degree $\geq d(n)$.
\item any 3-colorable graph on $n$ vertices  with maximum degree $\leq d(n)$.
\end{itemize}
Then we can make progress towards $\tO(f(n))$-coloring on any
3-colorable graph on $n$ vertices.
\end{lemma}
Using the SDP from \cite{KMS98}, we can make progress
towards $d(n)^{1/3}$ for graphs with degree below $d(n)$,
so by Lemma \ref{lem:sdp}, we
may assume
\begin{equation}\label{eq:min-degree-sdp}
k\leq\ds^{1/3}.
\end{equation}
We can do even better using the strongest SDP result of Chlamtac
from \cite{Chl07}:
\begin{theorem}[{\cite[Theorem 15]{Chl07}}]\label{thm:eden}
For any $\tau>\frac{6}{11}$ there is a $c>0$ such that
there is a polynomial time algorithm that for any 3-colorable graph
$G$ with $n$ vertices and all degrees below $\Delta = n^\tau$ finds an
independent set of size $\tOmega\left(n/\Delta^{1/(3+3c)}\right)$.
Hence we can make Type 1 progress towards an $\tO\left(\Delta^{1/(3+3c)}\right)=
\tO\left(n^{\tau/(3+3c)}\right)$-coloring.

The requirement on $\tau$ and
$c$ is that $c< 1/2$ and $\lambda_{c,\tau}(\alpha)
= 7/3 + c +\alpha^2/(1-\alpha^2) - (1+c)/\tau - (\sqrt{(1+\alpha)/2}
+\sqrt{c(1-\alpha)/2})^2$ is positive for all $\alpha\in
[0,\frac{c}{1+c}]$.
\end{theorem}


\paragraph*{Two-level neighborhood structure}
The most complex ingredient we get from Blum \cite{Blum94} is a
certain regular second neighborhood structure.  Let $\ds$ be
the smallest degree in the graph $G$. In fact, we shall use the
slightly modified version described in \cite{KT17}.

Unless other progress is made, for some $\ds_1=\tOmega(\ds)$,
in polynomial time \cite{Blum94,KT17},
we can identify a 2-level neighborhood structure $H_1=(r_0,S_1,T_1)$ in $G$
consisting of:
\begin{itemize}\packlist
\item A root vertex $r_0$. We assume $r_0$ is colored red in any
3-coloring.
\item A first neighborhood $S_1\subseteq N(r_0)$ of size at least $\ds_1$.
\item A second neighborhood $T_1\subseteq N(S_1)$ of size at most $n/k$.
The sets $S_1$ and $T_1$ may overlap.
\item The edges between vertices in $H_1$ are the same as those in $G$.
\item The vertices in $S_1$ all have degrees at least $\ds_1$ into $T_1$.
\item For some $\dt_1$ the degrees from $T_1$ to $S_1$ are all between $\dt_1$ and
$5\dt_1$.
\end{itemize}

\section{Recursive combinatorial coloring}\label{sec:focs12}

Our algorithm follows the same pattern as that in \cite{KT17}, but adds in certain side-cuts
leading us to both stronger and cleaner bounds.
Below we describe the algorithm, emphasizing the novel additions.

First, we will be able to have different and stronger parameters than those in \cite{KT17}. Given a 3-colorable graph with
minimum degree $\ds\geq\sqrt n$, we will make progress 
towards a $k$-coloring for
\begin{equation}\label{eq:k}k
=2^{\,(\log\log n)^2}\sqrt{n/\ds}
\end{equation}
Since $\nh=n/k^2$, this is equivalent to 
\begin{equation}\label{eq:delta-nh}
\ds/\nh=4^{\,(\log\log n)^2}.
\end{equation}
We will use the above 2-level neighborhood
structure $H_1=(r_0,S_1,T_1)$, and we are
going to recurse on induced subproblems
$(S,T)\subseteq (S_1,T_1)$ defined in terms of subsets $S\subseteq
S_1$ and $T\subseteq T_1$. The edges considered in the subproblem
are exactly those between $S$ and $T$ in $G$. This
edge set is denoted $E(S,T)$. 

A pseudo-code for our whole algorithm is presented in Algorithm \ref{alg:main}.
\begin{algorithm}
\caption{Seeking progress towards $\tO(k)$ coloring}\label{alg:main}
\begin{small}
let $(S_1,T_1,\ds_1,\dt_1)$ be the initial two-level structure from Section \ref{sec:blum}\;
\For(\hspace{2em}// outer loop, round $j$){$j\asgn 1,2,...$}{
$(S,T)\asgn(S_j,T_j)$\;
\Repeat(\hspace{2em}// inner loop. Iterative version of recursive
{$\monochromatic(S_j,T_j)$}){$|E(S,T)|< \dt_j|T|/2$}{
\lIf{$|S|\leq 1$}{\Return ``Error A''}
$U\asgn \{v\in T\,|\,d_S(v)\geq \dt_j/4\}$\;
\lIf{$|U|<\nh$}{\Return ``Error B''}
check $U$ multichromatic with Lemma \ref{lem:blum};\hspace{2em} // if not, progress was found and we are done\\
\If{ $\exists v\in U$ such that {\em $\cutcolornew(S,T,v)$} returns ``sparse cut around $(X(v),Y(v))$''}
{$(S,T)\asgn (X,Y)$} 
\lElse{\Return ``$S$ is monochromatic in every 3-coloring,
so monochromatic progress found''
}
}
$(X',Y')\asgn\sidecut(S,T,S_j,T_j)$\quad //NEW\\
\lIf{$|Y'|<|T|$}{$(S,T)\asgn (X',Y')$\quad //NEW}
$(S_j,T_j,\ds_j,\dt_j)=\regul(S,T)$\;
}
\end{small}
\end{algorithm}
Note that the
algorithm has two error events: Error A and Error B. We will make sure they never happen, and that the algorithm terminates, implying that
it does end up making progress. 

As in  \cite{KT17}, our
algorithm will have two loops; the inner and the outer. Both are combinatorial. 
When we
start outer loop $j$ it
is with a quadruple
$(S_j,T_j,\ds_j,\dt_j)$
where $(S_j,T_j)\subseteq (S_1,T_1)$ is {\em regular\/} in the sense that:
\begin{itemize}\packlist
\item The degrees from $S_j$ to $T_j$ are at least $\ds_j$.
\item The degrees from $T_j$ to $S_j$ are between $\dt_j$ and $5\dt_j$.
\end{itemize}
We note that the root vertex $r_0$ is not changed, and we will always view it as red in an unknown red-green-blue coloring.
In addition to regularity, we have the following two pre-conditions:
\begin{eqnarray}
\ds_j&=&\ds/(\log n)^{O(j)}=\omega(\nh).\label{eq:dsj}\\
\dt_j&\geq&4\ds_j/\nh.\label{eq:dtj}
\end{eqnarray}
Proving these pre-conditions is a non-trivial part of our later analysis.

\subsection{Inner loop}
The first thing we do in iteration $j$ of the inner loop is that we enter the inner loop which is almost
identical to that in \cite{KT17}.
Within the inner loop, we say that a vertex in $T$ has {\em high $S$-degree\/} if
its degree to $S$ is bigger than
$\dt_j/4$, and we will make sure that any subproblem $(S,T)$ considered
satisfies:
\begin{invariants}
\item\label{inv:high-degrees} We have at least $\nh$ vertices of high $S$-degree
in $T$.
\end{invariants}
This invariant implies that Error B never happens.

\subsubsection{Cut-or-color}
The most interesting part of
the inner loop is the 
subroutine $\cutcolor(S,T,t)$
which is taken from \cite{KT17} except for a small change that will be important to us. The input to the subroutine is a problem $(S,T)\subseteq (S_j,T_j)$ that starts with an arbitrary high $S$-degree vertex $t\in T$. It has one of the following outcomes:
\begin{itemize}\packlist
\item Some progress toward a $\tO(k)$-coloring. Then we are done,
so we assume that this does not happen.
\item A guarantee that if $r_0$ and $t$ have different colors in
a 3-coloring $C_3$ of $G$, then $S$ is monochromatic in $C_3$.
\item Reporting a ``sparse cut around a subproblem $(X,Y)\subseteq(S,T)$''
satisfying the following conditions:
\begin{invariants}
\item\label{inv:high-start} The  original high $S$-degree
vertex $t$ has all its neighbors to $S$ in $X$, that is,
$N_S(t)\subseteq X$.
\item\label{inv:XtoT-Y} All edges from $X$ to $T$ go to $Y$, so
there are no edges between $X$ and $T\setminus Y$.
\item\label{inv:X-extension} Each vertex $s'\in S\setminus X$ has $|N_Y(s')|< \nh$.
\item\label{inv:Y-extension} Each vertex  $t'\in T\setminus Y$ has $|N_Y(N_{N(r_0)}(t'))|< \nh$.
\end{invariants}
\end{itemize}
We note that in
\cite{KT17}, \ref{inv:Y-extension} only requires 
$|N_Y(N_{S}(t'))|< \nh$. Since
$S\subseteq N(r_0)$, this is
a weaker requirement. This
is why we say that our cuts
are sparser. A pseudo code for our revised \cutcolor\ is presented in Algorithm \ref{alg:cutcolornew}.

\begin{algorithm}\label{alg:cutcolornew}
\caption{$\cutcolornew(S,T,t)$}
$X=N_S(t)$; $Y=N_T(X)$;\\
\Loop{
\If{$X=S$}{\Return ``$S$ is monochromatic in every 3-coloring where $t$ and $r_0$ have different colors''}
\ElseIf(// $X$-extensions){there is $s\in S\setminus X$ such that $|N_{Y}(s)| \geq\nh$}
{check that $N_Y (s)$ is multichromatic in $G$ with Lemma~\ref{lem:blum}\;
add $s$ to $X$}
\ElseIf(// $Y$-extension){there is $t'\in T\setminus Y$ with $|N_Y(N_{N(r_0)} (t'))|\geq\nh$}
{check that $N_Y(N_{N(r_0)}(t'))$ is multichromatic in $G$ with Lemma~\ref{lem:blum}\;
add $t'$ to $Y$}
\Else(// $X\neq S$ and neither an $X$-extension nor a $Y$-extension is possible){
$(X(t), Y(t))\gets(X, Y)$\;
\Return ``sparse cut around $(X(t),Y(t))$''}
}
\end{algorithm}

Below we describe how  \cutcolornew\ works so as to satisfy the invariants, including our revision.
Consider any 3-coloring $C_3$ of $G$. 
$C_3$ is not known to the
algorithm. But, if the algorithm can guarantee that $S$ is monochromatic in $C_3$, then it can correctly declare that ``$S$ is monochromatic in
every 3-coloring where $t$ and $r_0$ have different colors''. Recall that 
$r_0$ is red and assume that $t$ is green in $C_3$. The last color is blue.

The first part of
$\cutcolor$ is essentially the coloring that Blum \cite[\S
  5.2]{Blum94} uses for dense bipartite graphs. Specifically, let $X$ be the neighborhood of $t$ in $S$ and let $Y$ be the
neighborhood of $X$ in $T$. As in \cite{Blum94} we note that all vertices of
$X$ must be blue, and that no vertex in $Y$ can be blue. We are going
to expand $X\subseteq S$ and $Y\subseteq T$ preserving the following
invariant:
\begin{invariants}
\item\label{inv:XY} if $r_0$ was red and $t$ was green in $C_3$ then
$X$ would be all blue and $Y$ would have no blue.
\end{invariants}
If we end up with $X=S$, then invariant \ref{inv:XY} implies that $S$ is
monochromatic in any 3-coloring where $r_0$ and $t$ have different colors.

\paragraph{$\bm X$-extension}
Now consider any vertex $s\in S$ whose degree into $Y$ is at least $\nh$.
Using Lemma \ref{lem:blum} we can check that $N_Y(s)$ is
multichromatic in $G$. Since $Y\supseteq N_Y(s)$ has no blue, we conclude that
$N_Y(s)$ is red and green, hence that $s$ is blue.
Note conversely that if $s$ was green, then all its neighbors in $Y$
would have to be red, and then the multichromatic test from
Lemma \ref{lem:blum} would have made progress.
Preserving invariant \ref{inv:XY}, we now add the blue $s$ to $X$ and all neighbors of $s$ in $T$ to
$Y$.  We shall refer to this as an {\em $X$-extension}.

\medskip

We now introduce $Y$-extensions. 
The important point will be that if we do not end
up with $X=S$, and if neither extension is possible, then we have 
a ``sparse cut'' around $(X,Y)$ that we can use for recursion. 

\paragraph{$\bm{Y}$-extension}
Consider a vertex $t'$ from $T\setminus Y$. Let $X'=N_{N(r_0)}(t')$ be its neighborhood in $N(r_0)$. Suppose $|N_Y(X')|\geq\nh$.  Using Lemma
\ref{lem:blum} we check that $N_Y(X')$ is multichromatic in $G$. We
now claim that $t'$ cannot be blue.  Suppose it was; then its
neighborhood has no blue and $N(r_0)$ is only blue and green, so
$X'=N_{N(r_0)}(t')$ must be all green.  Then the neighborhood of $X'$ has no
green, but $Y$ has no blue, so $N_Y(X')$ must be all red,
contradicting that $N_Y(X')$ is multichromatic. We conclude that $t'$
is not blue. Preserving invariant \ref{inv:XY}, we now add $t'$ to $Y$.

\subsubsection{Recursion towards a monochromatic set}

Assuming $\cutcolor$ above, we now review the main recursive algorithm from \cite{KT17}. Our inner loop in Algorithm \ref{alg:main} starts
with what corresponds to an iterative version of
the subroutine $\monochromatic$ from \cite{KT17} which takes as input a subproblem $(S,T)$ with $|S|>1$; otherwise
we get Error A. In the first
round of the inner loop,
we have $(S,T)=(S_j,T_j)$.

\drop{
The pseudo-code is presented in Algorithm \ref{alg:mono}.
\begin{algorithm}
\caption{$\monochromatic(S,T)$}\label{alg:mono}
\begin{small}
let $U$ be the set of high $S$-degree vertices in $T$\;
check that $U$ is multichromatic in $G$ with Lemma \ref{lem:blum};\hspace{2em} // if not, progress found and we are done\\
\If{ there is a $t\in U$ such that
{\em $\cutcolor(S,T,t)$} returns ``sparse cut around $(X,Y)$''}
{recursively call $\monochromatic(X,Y)$}
\Else{\Return ``$S$ is monochromatic in every 3-coloring''}\medskip
\end{small}
\end{algorithm}
}

Let $U$ be the set of high $S$-degree vertices in $T$. By \ref{inv:high-degrees} we have
$|U|\geq\nh$, so we can apply Blum's multichromatic test from Lemma
\ref{lem:blum} to $U$ in $G$.  Assuming that we did not make progress, we know that $U$ is
multichromatic in every valid 3-coloring.  We now apply $\cutcolor$ to
each $t\in U$, stopping only if a sparse cut is found or progress is
made.  If we make progress, we are done, so assume that this does not
happen. If a sparse cut around a subproblem $(X,Y)$ is found, we
recurse on $(S,T)=(X,Y)$.
\drop{\begin{lemma}
    For each $s'\in S$, if $N_T(s')\subseteq Y$, then
    $s'\in X$.
\end{lemma}
\begin{proof}
By \ref{inv:XtoT-Y} applied
recursively, we know that $s'$ has preserved all its neighbors from $T_j$, so $|N_{T}(s')|\geq \ds_j$. However, by \req{eq:dsj} and \req{eq:k-nh}, $\ds_j=\omega(\nh)$, so by \ref{inv:X-extension}, $s'\in X$.
\end{proof}}

The most interesting case is when we get neither progress nor a sparse cut. Here is the important result in \cite{KT17}. 
\begin{lemma}\label{lem:mono} If {\em $\cutcolor$} does not find progress
nor a sparse cut for any high $S$-degree $t\in U$, then
$S$ is monochromatic in every 3-coloring of $G$.
\end{lemma}
Thus, unless other progress is made, or we stop for other reasons, we end up
with a non-trivial set $S$ that is monochromatic in every 3-coloring, and
then monochromatic progress can be made. However, the correctness demands
that we respect \ref{inv:high-degrees} and only proceed
with a subproblem $(S,T)$ where $T$ has more than $\nh$ high $S$-degree
vertices (otherwise Lemma \ref{lem:blum} cannot be applied to $U$).

As proved in \cite{KT17}, invariant \ref{inv:high-degrees} must be
satisfied
if the average degree from $Y$ to $X$ is at least $\dt_j/2$. The proof
exploits pre-condition \req{eq:dsj} and that the
maximal degree $Y$ to $S$ is at most $5\dt_j$. 
If the average degree from $Y$ to $X$ drops below $\dt_j/2$, we terminate the
inner loop.

This completes our description of the inner loop in Algorithm \ref{alg:main}.
If we have not terminated
with progress, then the final
sparse cut $(X,Y)=(S,T)$
has average degree at most $\dt_j/2$ from $Y$ to $X$.

\paragraph{Regularization}
For now, skipping our new side cuts, we finish iteration $j$ of the outer loop as in \cite{KT17} by ``regularizing'' the degrees of vertices. 
The regularization is described in Algorithm \ref{alg:regul}, and it
is, in itself, fairly standard. Blum \cite{Blum94} used several similar
regularizations. In \cite{KT17} the following lemma is shown. 

\begin{algorithm}
\caption{$\regul(S,T)$}\label{alg:regul}
Let $d_\ell=(4/3)^\ell$\;
Partition the vertices of $T$ into sets
$U_\ell=\{v\in T\,|\,d_S(v)\in [d_\ell,d_{\ell+1})\}$\;
Subject to $d_\ell\geq \avg d_S(T)/2$ let $\ell$ maximize  $|E(U_\ell,S)|$\;
$\dt^{\,r}\asgn d_\ell/4$; $\Delta^r\asgn \avg d_{U_\ell}(S)/4$\;
Repeatedly remove vertices $v\in S$ with $d_{U_\ell}(v)\leq \ds^{\,r}$
and $w\in U_\ell$ with $d_S(w)\leq \dt^{\,r}$\;
$S^{\,r}\asgn S$; $T^{\,r}\asgn U_\ell$\;
\Return $(S^{\,r},T^{\,r},\ds^{\,r},\dt^{\,r})$
\end{algorithm}
\begin{lemma}\label{lem:regul} 
When {\em \regul}$(S,T)$ in Algorithm \ref{alg:regul} returns $(S^{\,r},T^{\,r},\ds^{\,r},\dt^{\,r})$
then  $\Delta^r\geq \avg d_T(S)/(30\log n)$ and $\dt^{\,r}\geq \avg d_S(T)/8$.
The sets $S^{\,r}\subseteq S$ and $T^{\,r}\subseteq T$ are both non-empty.
The degrees from $S^{\,r}$ to $T^{\,r}$ are at least $\ds^{\,r}$ and
the degrees from $T^{\,r}$ to $S^{\,r}$ are
between $\dt^{\,r}$ and $5\dt^{\,r}$.
\end{lemma}


To complete our explanation of Algorithm \ref{alg:main}, we have to describe our new side cuts, which is done in the next section.

\section{Introducing side cuts}
In this section, we will describe our new ''side cuts''
that can be used as an alternative to the sparse
cuts identified by
the inner loop. In outer round $j$, at
the end of the inner
loop, going through
nested sparse cuts, we
have got to the last
sparse cut $(X,Y)\subseteq (S_j,T_j)$ such
that $|E(X,Y)|<\delta_j|Y|/2$.

Now we are going to identify
a family of side cuts, one for
each $u\in Y$ with $d_{S_j\setminus X}(u)\geq \delta_j/3$. The
{\em side cut\/} $(X'(u),Y'(u))$ is
defined as follows.
\begin{itemize}
    \item $X'(u)=N_{S_j}(u)\setminus X$.
    \item $Y'(u)=N_{T_j}(X'(u))\setminus Y$.
\end{itemize}
Note that the above side cut
is disjoint from the sparse cut $(X,Y)$. 
See Figure \ref{fig1} for this intuition. 

The {\em best side cut\/} $(X'(u),Y'(u))$
is the one with the smallest
$Y'(u)$. Algorithm \ref{alg:inner-new} finds
the best side cut $(X',Y')$
and in Algorithm \ref{alg:main}, it replaces 
the sparse cut $(X,Y)$ if
$Y'$ is smaller than $Y$.

\begin{algorithm}
\caption{$\sidecut(X,Y,S_j, T_j)$}\label{alg:inner-new}
$(X',Y')\asgn (S_j,T_j)$\;
\For{$u\in Y$}{
\If{$d_{S_j\setminus X}(u)\geq\delta_j/3$}{
$X'(u)\asgn N_{S_j}(u)\setminus X$\;
$Y'(u)\asgn N_{T_j}(X'(u))\setminus Y$\;
\lIf{$|Y'(u)|<|Y'|$}{$(X',Y')\asgn
(X'(u),Y'(u))$}
}
}
\Return $(X',Y')$
\end{algorithm}

Incidentally, we note that
the final sparse cut $(X,Y)$ found in external round $j$ is also the sparse cut with the smallest $Y$, so the $(X,Y)$ we end up using is the
one minimizing $Y$ among
all sparse cuts and side cuts considered. 

In Algorithm \ref{alg:main}, the best sparse or side cut $(X,Y)$ gets assigned to $(S,T)$ before
we regularize it, obtaining
the new quadruple $(S_{j+1},T_{j+1},\ds_{j+1},\dt_{j+1}).$  We shall refer to this best cut found in outer iteration $j$ as $(X_j,Y_j)$. This completes our description of Algorithm \ref{alg:main}.

\begin{figure}
\centering
\includegraphics[height=6cm]{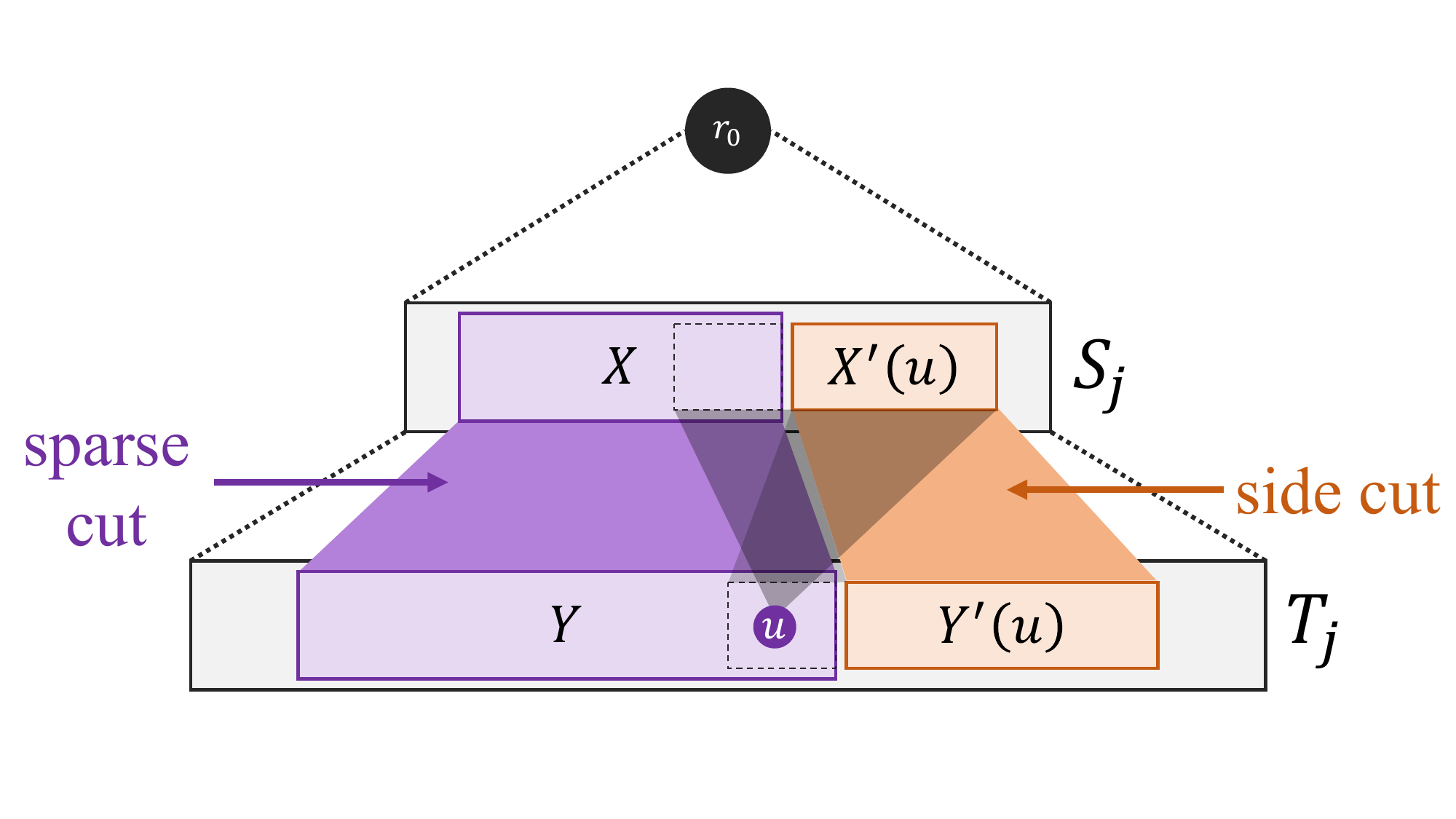}
    \caption{Side cuts and sparse cuts}
    \label{fig1}
\end{figure}

We note that our new side cuts are much simpler than the sparse cuts from \cite{KT17}. The interesting thing is that it makes a substantial difference when we take the best of the two. This is proved in the rest of this paper.

\section{Analyzing outer round $j$ including side cuts}
We are now going to analyze
outer round $j$. The novelty relative to \cite{KT17} is the impact
of side cuts together with the small tightening of \cutcolor. Our main result for iteration $j$ will be
that 
\[
|Y_j|\leq O(\sqrt{\nh |T_j|}).\]
Since $|T_{j+1}|\leq |Y_j|$, this bound can be
used recursively.

\paragraph{Cuts.}
We claim that any sparse or
side cut $(X,Y)$ considered
will satisfy
\begin{eqnarray}
\min d_Y(X)&\geq&\ds_j/2\label{eq:min-ds}\\
|X|&\geq& \dt_j/4\label{eq:lowerX}\\
|Y|&\geq& \ds_j^2/(8\nh)\label{eq:lowerY}
\end{eqnarray}
In \cite{KT17}, this was 
already proved for the sparse
cuts $(X,Y)$. More precisely, for sparse cuts, vertices in $X$ preserve all their neighbors in $T_j$. This
follows recursively from Invariant (iii), and therefore $\min d_Y(X)\geq\ds_j$, implying \req{eq:min-ds}. Also \req{eq:lowerX} follows because $X$ contains the
neighborhood of a high $S$-degree vertex, which means
degree at least $\dt_j/4$. Finally,  \cite[Eq. (17)]{KT17} states
$|Y|\geq\ds_j^2/(2\nh)$ implying \req{eq:lowerY}. 

\begin{lemma}\label{lem2}
Each side cut $(X'(u),Y'(u))$
satisfies conditions
in \req{eq:min-ds}--\req{eq:lowerY}.
\end{lemma}
\begin{proof}
First, we note that \req{eq:lowerX} holds because the side cut is only defined if $d_{S_j\setminus X}(u)=|X'(u)|\geq \delta_j/3$.

We now observe that any vertex $v' \in S\setminus X$ has at most $\nh$ neighbors in $Y$ by invariant \ref{inv:X-extension}. Moreover,  by \req{eq:dsj},
$\ds_j=\tomega(\nh)$, so the degree from $v' \in X'(u)$ to $Y'(u)$ is at least $\ds_j-\nh>\ds_j/2$, which implies \req{eq:min-ds}. 

For \req{eq:lowerY}, 
again by \req{eq:min-ds}
the degree from any vertex in $X'(u)$ to $Y'(u)$ is at least 
$\ds_j/2$.  
By Lemma~\ref{lem:large-neighborhoods}, 
$|Y'(u)|\geq \min\{(\ds_j/2)/\nh,|X'(u)|\}\,\ds_j/4$. 
We have $|X'(u)| \geq \dt_j/3$ and
by pre-condition \req{eq:dtj}, we
have $\dt_j\geq 4\ds_j/\nh$, so 
$|X'(u)| \geq (4\ds_j/3)/\nh$. Therefore the min evaluates
to the first option, so we get
$|Y'(u)|\geq \frac{\ds_j^2}{8\nh}$, as required in \req{eq:lowerY}. 
\end{proof}

\paragraph{Final sparse cut.} 
In the rest of this section, let $(X,Y)$ 
denote the last sparse cut considered in outer iteration $j$. The best side cut is then found by calling $\sidecut(X,Y,S_j,T_j)$.

Let $Y'' \subseteq Y$ be 
the set of vertices $u\in Y$ such that at least $\dt_j/3$ of the
$S_j$-neighbors of $u$ are in $S_j \backslash X$. Then $Y''$ are exactly the vertices
for which side cuts $(X'(u),Y'(u))$ are defined.
\begin{lemma}\label{lem:lem31}
We have at least $\dt_j|Y|/6$ edges from $Y''$ to  $S_j \backslash X$. 
\end{lemma}
\begin{proof}
The loop
over sparse cuts in Algorithm \ref{alg:main} only terminates because $|E(S,T)|< \dt_j|T|/2$ where
$(S,T)$ is our $(X,Y)$. Therefore we must have $|E(Y,S_j\setminus X)|\geq \dt_j|Y|/2$. 

By definition,
the vertices
from $Y\setminus Y''$ have less than 
$\dt_j/3$ edges to $S_j\setminus X$ amounting to a total of
less than $|Y|\dt_j/3$ edges to $S_j\setminus X$. Thus
we must have at least $\dt_j|Y|/2-|Y|\dt_j/3=\dt_j|Y|/6$
edges from $Y''$ to $S_j\setminus X$.
\end{proof}


We define 
\[\mu=\frac{|Y_j|\nh}{\ds_j^2}.\] 
Since $Y_j$ was a $Y$ from a sparse or side cut $(X,Y)$,
\req{eq:lowerY} implies $\mu\geq 1/8$.  We now show the following key lemma, which relates to Lemma 13 in \cite{KT17} (that is the most important in their analysis). 

\begin{lemma}\label{lem:lem3}
The number of edges from $Y''$ to
$S_j \backslash X$ is at most
\begin{equation}\label{eq:loss}
|T_j\setminus Y|\,\frac{5\dt_j\nh^2}{\mu \ds_j^2
  }.
\end{equation}
\end{lemma}
\begin{proof}
We have
\begin{align*}
\sum_{y\in Y} |N_{T_j}(N_{S_j}(y))\setminus Y|&=
|\{y\in Y, u\in {T_j}\setminus Y\mid N_{S_j}(y)\cap N_{S_j}(u)\neq\emptyset\}|\\
&=\sum_{u\in {T_j}\setminus Y} |N_Y(N_{{S_j}}(u))|.
\end{align*}
Consider any $u\in {T_j}\setminus Y$. During
the inner loop, there was a first
sparse cut $(X',Y')\supseteq (X,Y)$ with
$u\not\in Y'$, and then, by Invariant \ref{inv:Y-extension}, we
had 
$|N_{Y'}(N_{N(r_0)}(u))|<\nh$. 
Since
$Y'\supseteq Y$ and $N(r_0)\supseteq S_j$, this implies 
$|N_Y(N_{S_j}(u))|<\nh$. 
Thus we have proved
\[\sum_{y\in Y} |N_{T_j}(N_{S_j}(y))\setminus Y|\leq \,|{T_j}\setminus Y|\nh.\]


We will now relate the above bound to our side cuts $(X'(y), Y'(y))$ for $y\in Y''$.
By definition,  
\[Y'(y)=N_{T_j}(N_{S_j}(y)\setminus X)\setminus Y\subseteq
N_{T_j}(N_{S_j}(y))\setminus Y.\] 
Moreover,
$Y''\subseteq Y$, so we get
\begin{align}
\sum_{y\in Y''} |Y'(y)| \leq |{T_j}\setminus Y| \nh.\label{eq:double-sum'}
\end{align}
Recall that we defined
$\mu$ such that \[\mu \times \frac{\ds_j^2 }{\nh}=|Y_j|\]
and $|Y_j| \leq |Y'(y)|$ for all $y\in Y''$.

In combination with 
\req{eq:double-sum'},
this implies that
\[|Y''|\leq 
\frac{|{T_j}\setminus Y| \nh}{
\mu \times \frac{\ds_j^2 }{\nh} }= |{T_j}\setminus Y|\frac{\nh^2}{
\mu \ds_j^2}.
\]
We also
know that all degrees from $T_j$ to $S_j$ are bounded by $5\dt_j$
and this bounds the size of $X'(y)=N_{S_j}(y) \setminus X$. We now get the desired bound on
the number of edges from $Y''$ to ${S_j} \backslash X$ as
\begin{align*}
\sum_{y\in Y''} |N_{S_j}(y) \setminus X|&\leq
5\dt_j|Y''|\leq
5\dt_j |T_j\setminus Y|\frac{\nh^2}{
\mu \ds_j^2}=|{T_j}\setminus Y|\,\frac{5\dt_j \nh^2}{\mu \ds_j^2
 }.\\[-8ex]\end{align*}
\end{proof}

The number of edges from $Y''$ to
$S_j \backslash X$ are bounded from below 
by $\dt_j|Y|/6$ by
Lemma \ref{lem:lem31} and from above 
by 
$|T_j\setminus Y|\,\frac{5\dt_j\nh^2}{\mu \ds_j^2
  }$ by
Lemma \ref{lem:lem3},
so we get
\begin{equation}
|Y| \leq |{T_j}|(30\nh^2)/(\mu \ds_j^2).\label{eq:upperY}
\end{equation}
Combined with $\mu\times \frac{\Delta_j^2}\nh=|Y_j|\leq|Y|$, 
we get an upper bound on $\mu$:
\begin{equation}
\mu \leq \frac{\nh \sqrt{30 \nh|{T_j}|}}{\ds_j^2 }.\label{eq:upperYnew}
\end{equation}
We now conclude with the main result from our analysis of our new inner loop:
\begin{lemma}\label{lem:upperYnew21}
\begin{equation}
|Y_j|=\mu\times \frac{\Delta_j^2}\nh\leq \sqrt{30 \nh|T_j|}.\label{eq:upperYnew21}
\end{equation}
\end{lemma}
Note that if \req{eq:upperYnew21} is smaller than the lower bound $\ds_j^2/(8\nh)$ from \req{eq:lowerY}, then progress must have been made before the round finished and this is what are ultimately hoping for.

\paragraph{Comparison to \cite{KT17}.}
In \cite{KT17}, there are no side cuts and no $\mu$, so they only had an
upper bound corresponding to \req{eq:upperY}
with $\mu\geq 1/8$, that is,
\begin{equation}
|Y_j| \leq O(|{T_j}|(\nh/\ds_j)^2).\label{eq:old-Y}  
\end{equation}
To appreciate the difference, consider the first
round where we only have $|T_1|\leq n/k$. Then
\req{eq:upperYnew21} yields $|Y_1|=O(n/k^{3/2})$, gaining a factor $\sqrt k$ over $n/k$. For comparison, with \req{eq:old-Y},
we gain only a factor $(\nh/\ds_j)^2$ which will be subpolynomial by \req{eq:delta-nh}. 

\paragraph{Proceed to the outer loop} 
Before going to the outer loop, we give requirements to satisfy during the whole outer loop, which is really the same as that in \cite{KT17}.

We say round $j$ is {\em good\/} if
\begin{itemize}
\item the pre-conditions \req{eq:dsj} and \req{eq:dtj} are satisfied at
the beginning of the round.
\item No error is made during the round.
\item \req{eq:min-ds}--\req{eq:lowerY} and \req{eq:upperYnew21} are satisfied as long as
no progress is made.
\end{itemize}
Because both our side cuts
and the original sparse cuts from \cite{KT17} satisfy conditions
in \req{eq:min-ds}--\req{eq:lowerY}, a simple generalization of the
analysis from \cite{KT17} implies:
\begin{lemma}[{\cite{KT17}}]\label{lem:ind-step}
Round $j$ is good if and only if pre-conditions \req{eq:dsj} and \req{eq:dtj} are satisfied.
\end{lemma}
The main difference is that our new bound \req{eq:upperYnew21} makes it much easier to satisfy the pre-conditions.

\section{Analysis of outer loop}\label{sec:outer-loop}

Following the pattern in \cite{KT17}, but using our new Lemma \ref{lem:ind-step}, we will prove an inductive statement that implies that the outer loop in Algorithm \ref{alg:main} continues with no errors until we make progress with a good
coloring. The result, stated below, is both 
simpler and stronger
than that in \cite{KT17}.
\begin{theorem}\label{thm:restrictions}
Consider a 3-colorable graph with minimum degree $\ds\geq\sqrt n$. Let 
\begin{equation}\label{eq:c-bound}
k=2^{(\log\log n)^2}\sqrt{n/\ds}.
\end{equation}
Algorithm \ref{alg:main} will make only good rounds, and make progress
towards an $\tO(k)$ coloring no later than round $\lfloor \log\log n\rfloor$.
\end{theorem}
\begin{proof}
We are going to analyze round $j\leq c=\lfloor \log\log n\rfloor$ of the outer loop.
Assuming that all previous
rounds have been good,
but no progress has been made, we will show that 
the preconditions of
round $j$ are satisfied,
hence that round $j$ must
also be good. Later we will also show that progress must be made no later than
round $j=c$.

To show that the preconditions
are satisfied, we will develop
bounds for $|T_j|$, $\ds_j$,  and $\dt_j$ assuming that
all previous rounds have
been good and that no progress has been made. We
already know that $|T_1|\leq n/k$ and
$\ds_1=\tOmega(\ds)$. Moreover,
$5\dt_1\geq \avg d_{S_1}(T_1)\geq
\Delta_1|S_1|/|T_1|\geq \Delta_1^2k/n$, so $\dt_1=\tOmega(\ds^2 k/n)$.

Let us observe that since no progress is made, round $j$ ends up regularizing. As in
previous sections, we let $X_j$ and $Y_j$
denote the last values of $S$ and $T$ in round $j$. Then
\[(S_{j+1},T_{j+1},\ds_{j+1},\dt_{j+1})=\regul(X_j,Y_j).\]
We will now derive inductive bounds on $|T_{j+1}|$, $\ds_{j+1}$, and $\dt_{j+1}$.
\paragraph{Computing $\ds_j$.}
By Lemma\tref{lem:regul} and \req{eq:min-ds},
we get 
\[\ds_{j+1}\geq \avg d_{Y_j}(X_j)/(30\log n) 
\geq \ds_j/(60\log n)\]
Inductively, we conclude for
any $j$ that 
\begin{equation}
\ds_j\geq \ds_1/(60\log n)^{j-1}=\tOmega(\ds/(60\log n)^j).\label{eq:ds-ind}
\end{equation}
We note that $\ds_j$ is 
just a lower bound on the degrees from $S_j$ to $T_j$ and we can assume that it 
is decreasing like in \req{eq:ds-ind}.

\paragraph{Satisfying pre-condition 
\req{eq:dsj} stating $\nh=o(\ds_{j})$.}
By \req{eq:c-bound}
we have 
\begin{equation}\label{eq:ds-nh}
\ds/\nh=\ds k^2/n=4^{(\log\log n)^2}.
\end{equation}
For $j< \log\log n$, we have
\[(60\log n)^j=\tilde o(4^{(\log\log n)^2})\textnormal,\]
so we get \req{eq:dsj} in the stronger from
\begin{equation}\label{eq:dsj*}
\nh=\tilde o(\ds_j).
\end{equation}

\paragraph{Computing $|T_j|$.}
When we start on
the outer loop, we have
$|T_1|\leq n/k=\nh k$ and
the regularization implies
$T_{j+1}\subseteq Y_j$, so
\req{eq:upperYnew21} implies
$|T_{j+1}|\leq|Y_j|\leq \sqrt{30 \nh|T_j|}$.
Thus, assuming that
nothing goes wrong,
\begin{equation}\label{eq:newupperY}
|Y_j|<30\nh k^{1/2^j}.
\end{equation}

\paragraph{Computing $|\dt_j|$.}
Recall that $\dt_1=\tOmega(\ds^2k/n)$ and
consider $j>1$. We have
$\avg d_{X_j}(Y_j)\geq(\ds_j/2)|X_j|/|Y_j|$
by \req{eq:min-ds}, and
$|X_j|\geq \dt_j/4$ by \req{eq:lowerX},
and $|Y_j|<30\,\nh k^{1/2^j}$ by \req{eq:newupperY}. Therefore
\begin{eqnarray*}
\avg d_{X_j}(Y_j)=\Omega\left(\ds_j\dt_j/( \nh k^{1/2^j})\right).
\end{eqnarray*}
By Lemma \ref{lem:regul}
\[    \dt_{j+1}=\tOmega\!\left(\avg d_{X_j}(Y_j)\right)=\tOmega\!\left(\ds_j\dt_j/( \nh k^{1/2^j})\right)=\dt_j\, \tOmega\left(\ds_j /\nh\right)/k^{1/2^j}\]
By induction and since the $\Delta_j$ are assumed to be decreasing, we get
\begin{equation}\label{eq:dtj+1}
    \dt_{j+1}=\tOmega(\ds^2k/n)\, \tOmega\left(\ds_j /\nh\right)^j/ k^{1-1/2^j}
    =\tOmega\left(\ds_j /\nh\right)^j\ds^2k^{1/2^j}/n.
\end{equation}

\paragraph{Satisfying pre-condition 
\req{eq:dtj} stating $\dt_j\geq 4\ds_{j}/\nh$.}
We note that $\ds_j$ is 
just a lower bound on the degrees from $S_j$ to $T_j$
and we can assume that it 
is decreasing like in \req{eq:ds-ind}.
The proof is divided into three different cases: $j=1$, $j=2$, and $j>2$. For $j=1$, we have 
$\dt_1=\tOmega(\ds^2k/n)=(\ds/\nh)
\tOmega(\ds/k)$. Here $\tOmega(\ds/k)=\omega(1)$ by \req{eq:min-degree}, so pre-condition \req{eq:dtj} is satisfied. 

For larger
$j>1$, that is, to make it past the first round, we need lower bounds on the minimum degree $\ds$. 
For $j=2$, by 
\req{eq:dtj+1}, $\dt_2
    =\tOmega\left(\ds_1 /\nh\right)\ds^2k^{1/2}/n$.
However, since $\ds\geq\sqrt n$ and $k>\sqrt{n/\ds}$,
\[\ds^2k^{1/2}/n>\ds^2(n/\ds)^{1/4}/n\geq n^{1/8}.\]
Thus $\dt_2 >\tOmega\left(\ds_1 /\nh\right)n^{1/8}=\tomega(\ds_1/\nh)$.
Above we do have some slack in that it would have sufficed with $\ds>n^{3/7+\Omega(1)}$, but for later rounds, we cannot have 
$\ds$ much smaller than $\sqrt n$. More precisely, for $j> 2$,  by 
\req{eq:dtj+1}, 
\[\dt_{j+1}
    =\tOmega\left(\ds_j /\nh\right)^j\ds^2k^{1/2^j}/n\geq
    \tOmega\left(\ds_j /\nh\right)^2\ds^2k^{1/2^j}/n
    =\tomega(\ds_j /\nh)\,\ds^2k^{1/2^j}/n.
 \]
The last derivation follows from
\req{eq:dsj*} using $j> 2$. With $\ds\geq\sqrt n$, we have $\ds^2k^{1/2^j}/n>1$
 and \req{eq:dtj}  follows.

\paragraph{Progress must be made.}
Assuming only good rounds,
we will now argue that progress must be made no later
than round $\log\log n$. Assuming no
progress, recall from \req{eq:newupperY}
that 
\[|Y_j|<30\,\nh k^{1/2^j}.
\]
From \req{eq:lowerY}, we know that any $Y$ considered (without progress), including $Y_j$,
is of size at least $\ds_j^2/(8\nh)$.
Thus, if we make it through round $j$ without
progress, we must have
\[\ds_j^2/(8\nh)<30\,\nh k^{1/2^j}\iff
(\ds_j/\nh)^2<240\,k^{1/2^j}
\]
By \req{eq:dsj*}, $\ds_j/\nh=\tomega(1)$ for
$j< \log\log n$, and 
$k^{1/2^j}<2$ for $j=\lceil \log\log k \rceil<\log\log n$,
so progress must be made no later than 
round $j=c=\lfloor \log\log n \rfloor$.
This
completes our proof of Theorem \ref{thm:restrictions}. 
\end{proof}

\medskip

Using Theorem \ref{thm:main}, we can show the following: 
\begin{theorem}\label{thm:main2}
In polynomial time, we can color any 3-colorable $n$ vertex graph
using $\tO(n^{0.19747})$ colors.
\end{theorem}
\begin{proof}
As in \cite{KT17}, we use Chlamtac's SDP \cite{Chl07} for low degrees. By Theorem \ref{thm:eden}, for maximum degree below
$\ds=n^\tau$ with $\tau={0.605073}$
and $c=0.0213754$, we can make progress towards $k=\tO(n^{0.19747})$ coloring. By Lemma \ref{lem:sdp}, we may therefore assume that
$\ds$ is the minimum degree. This is
easily above $\sqrt n$, and then by Theorem \ref{thm:restrictions}, we get progress
towards 
\[\sqrt{n/\ds}2^{(\log\log n)^2}=
n^{(1-0.605073)/2}2^{(\log\log n)^2}<
n^{0.19747}.\]
Thus, in polynomial time, we  can color any 3-colorable graph with
$\tO(n^{0.19747})$ colors.
\end{proof}
\label{end}


\begin{thebibliography}{10}

\bibitem{ACC06}
S.~Arora, E.~Chlamtac, and M.~Charikar.
\newblock New approximation guarantee for chromatic number.
\newblock In {\em Proc. 38th STOC}, pages 215--224, 2006.

\bibitem{AG11}
S.~Arora and R.~Ge.
\newblock New tools for graph coloring.
\newblock In {\em Proc. APPROX-RANDOM}, pages 1--12, 2011.

\bibitem{ARV09}
S.~Arora, S.~Rao, and U.~Vazirani.
\newblock Expanders, geometric embeddings and graph partitioning.
\newblock {\em J. ACM}, 56(2):1--37, 2009.
\newblock Announced at STOC'04.

\bibitem{BR90}
B.~Berger and J.~Rompel.
\newblock A better performance guarantee for approximate graph coloring.
\newblock {\em Algorithmica}, 5(3):459--466, 1990.

\bibitem{Blum94}
A.~Blum.
\newblock New approximation algorithms for graph coloring.
\newblock {\em J. ACM}, 41(3):470--516, 1994.
\newblock Combines announcements from STOC'89 and FOCS'90.

\bibitem{BK97}
A.~Blum and D.R. Karger.
\newblock An {${\tilde O}(n^{3/14})$}-coloring algorithm for 3-colorable
  graphs.
\newblock {\em Inf. Process. Lett.}, 61(1):49--53, 1997.

\bibitem{Chl07}
E.~Chlamtac.
\newblock Approximation algorithms using hierarchies of semidefinite
  programming relaxations.
\newblock In {\em Proc. 48th FOCS}, pages 691--701, 2007.

\bibitem{DHSV2015}
I.~Dinur, P.~Harsha, S.~Srinivasan, and G.~Varma.
\newblock Derandomized graph product results using the low degree long code.
\newblock In {\em Proc. 32nd STACS}, pages 275--287, 2015.

\bibitem{DMR09}
I.~Dinur, E.~Mossel, and O.~Regev.
\newblock Conditional hardness for approximate coloring.
\newblock {\em SIAM J. Comput.}, 39(3):843--873, 2009.
\newblock Announced at STOC'06.

\bibitem{FK}
U.~Feige and J.~Kilian.
\newblock Zero-knowledge and the chromatic number.
\newblock {\em J. Comput. System Sci.}, 57:187--199, 1998.

\bibitem{FLS04}
U.~Feige, M.~Langberg, and G.~Schechtman.
\newblock Graphs with tiny vector chromatic numbers and huge chromatic numbers.
\newblock {\em SIAM J. Comput.}, 33(6):1338--1368, 2004.
\newblock Announced at FOCS'02.

\bibitem{GJS76}
M.R. Garey, D.S. Johnson, and L.J. Stockmeyer.
\newblock Some simplified np-complete graph problems.
\newblock {\em Theor. Comput. Sci.}, 1(3):237--267, 1976.
\newblock Announced at STOC'74.

\bibitem{GW95}
M.X. Goemans and D.P. Williamson.
\newblock Improved approximation algorithms for maximum cut and satisfiability
  problems using semidefinite programming.
\newblock {\em J. ACM}, 42(6):1115--1145, 1995.
\newblock Announced at STOC'94.

\bibitem{VGK04}
V.~Guruswami and S.~Khanna.
\newblock On the hardness of 4-coloring a 3-colorable graph.
\newblock {\em SIAM Journal on Discrete Mathematics}, 18(1):30--40, 2004.

\bibitem{GS20}
Venkatesan Guruswami and Sai Sandeep.
\newblock d-to-1 hardness of coloring 3-colorable graphs with {O(1)} colors.
\newblock In {\em 47th ICALP}, volume 168 of {\em LIPIcs}, pages 62:1--62:12,
  2020.

\bibitem{Ha}
J.~H{\aa{}}stad.
\newblock Clique is hard to approximate within $n^{1-\varepsilon}$.
\newblock {\em Acta Math.}, 182:105--142, 1999.

\bibitem{KMS98}
D.R. Karger, R.~Motwani, and M.~Sudan.
\newblock Approximate graph coloring by semidefinite programming.
\newblock {\em J. ACM}, 45(2):246--265, 1998.
\newblock Announced at FOCS'94.

\bibitem{Karp75}
R.~M. Karp.
\newblock On the computational complexity of combinatorial problems.
\newblock {\em Networks}, 5:45--68, 1975.

\bibitem{KT12}
K.~Kawarabayashi and M.~Thorup.
\newblock Combinatorial coloring of 3-colorable graphs.
\newblock In {\em Proc. 53rd FOCS}, pages 68--75, 2012.

\bibitem{KT17}
K.~Kawarabayashi and M.~Thorup.
\newblock Coloring 3-colorable graphs with less than
  \emph{n}\({}^{\mbox{1/5}}\) colors.
\newblock {\em J. {ACM}}, 64(1):4:1--4:23, 2017.
\newblock Announced at FOCS'12 and STACS'14.

\bibitem{KLS00}
S.~Khanna, N.~Linial, and S.~Safra.
\newblock On the hardness of approximating the chromatic number.
\newblock {\em Combinatorica}, 20(3):393--415, 2000.

\bibitem{Sze94}
M.~Szegedy.
\newblock A note on the $\theta$ number of {}{Lov\'asz} and the generalized
  {Delsarte} bound.
\newblock In {\em Proc. 35th FOCS}, pages 36--39, 1994.

\bibitem{Wig83}
A.~Wigderson.
\newblock Improving the performance guarantee for approximate graph coloring.
\newblock {\em J. ACM}, 30(4):729--735, 1983.
\newblock Announced at STOC'82.

\end{thebibliography}
\end{document}